\documentclass[twocolumn]{IEEEtran}
\usepackage{amsfonts}
\usepackage{cite,graphicx,amsmath,amsthm}
\usepackage{subfigure}
\usepackage{fancyhdr}
\usepackage{dsfont}
\usepackage{array,color}
\usepackage{bm}
\usepackage{float}
\usepackage{algorithm}
\usepackage{algpseudocode}
\usepackage{multirow}
\usepackage{booktabs}
\usepackage{multirow}
\usepackage{makecell}

\allowdisplaybreaks[4]

\newtheorem{theorem}{Theorem}

\newtheorem{lemma}{Lemma}

\newtheorem{proposition}{Proposition}

\newtheorem{corollary}{Corollary}

\newtheorem{property}{Property}

\newtheorem{remark}{Remark}

\newtheorem{claim}{Claim}

\begin{document}
	\title{\huge Large-Scale Bandwidth and Power Optimization for Multi-Modal Edge Intelligence Autonomous Driving}

	\author{Xinrao Li, Tong Zhang, \textit{Member, IEEE}, Shuai Wang, \textit{Member, IEEE}, Guangxu Zhu, \textit{Member, IEEE}, \\ Rui Wang, \textit{Member, IEEE}, and Tsung-Hui Chang \textit{Fellow, IEEE},
		\vspace{-15pt}
		\thanks{
			
			Xinrao Li, Tong Zhang, and Rui Wang are with the Department of Electrical and Electronic Engineering, Southern University of Science and Technology, Shenzhen 518055, China (e-mail: \{11930632, zhangt7, wang.r\}@sustech.edu.cn).
			
			Shuai Wang is with the Shenzhen Institute of Advanced Technology, Chinese Academy of Sciences, Shenzhen 518055, China (e-mail: s.wang@siat.ac.cn).
			
			Guangxu Zhu is with Shenzhen Research Institute of Big Data, Shenzhen 518172, China (gxzhu@sribd.cn).
			
			Tsung-Hui Chang is with the School of Science and Engineering, The Chinese University of Hong Kong Shenzhen and Shenzhen Research Institute of Big Data, 518172, China (e-mail: tsunghui.chang@ieee.org).
			
			%Corresponding Author: Tong Zhang and Rui Wang.
		}% <-this % stops a space
	}

	\maketitle

	\begin{abstract}
		Edge intelligence autonomous driving (EIAD) offers computing resources in autonomous vehicles for training deep neural networks. However, wireless channels between the edge server and the autonomous vehicles are time-varying due to the high-mobility of vehicles.	Moreover, the required number of training samples for different data modalities, e.g., images, point-clouds, is diverse. Consequently, when collecting these datasets from vehicles to the edge server, the associated bandwidth and power allocation across all data frames is a large-scale multi-modal optimization problem.
		This article proposes a highly computationally efficient algorithm that directly maximizes the quality of training (QoT).
		The key ingredients include a data-driven model for quantifying the priority of data modality and two first-order methods termed accelerated gradient projection and dual decomposition for low-complexity resource allocation.
		Finally, high-fidelity simulations in Car Learning to Act (CARLA) show that the proposed algorithm reduces the perception error by $3\%$ and the computation time by $98\%$.
	\end{abstract}

	
	\vspace{-9pt}
	\IEEEpeerreviewmaketitle
	\section{Introduction}
	
	Edge intelligence autonomous driving (EIAD) is a promising paradigm to ease the conflict between the resource-hungry model training and the resource-limited vehicle platforms \cite{zhangjun}.
	Compared to cloud-assisted approaches, EIAD achieves better privacy protection and lower latency by providing computing resources in close proximity to autonomous vehicles \cite{zhangjun}.
	Among others, model training is the most fundamental task in EIAD systems, which consists of dataset generation, transmission, calibration, annotation, and processing.
	However, since EIAD systems need to train an ensemble of deep neural networks (DNNs) for learning semantic, geometry, and motion representations, the datasets are multi-modal, and the trained DNNs are heterogeneous \cite{tits}. Therefore, the required number of training samples is diverse, disabling conventional throughput-oriented approaches.
	
	Another challenge of EIAD is the time-varying wireless channels between the edge server and the autonomous vehicles due to high mobility, which makes the coherent time (i.e., a unit of time block for resource allocation) very small \cite{44,48,49,51}.
	Hence given a common data volume of an AD dataset, e.g., 100 GB, the number of time blocks for transmission can be very large.
	Consequently, EIAD systems require a fast large-scale optimizer for wireless resource allocation, and conventional convex optimization methods, e.g., the interior point method, are no longer suitable.
	
	In this article, we would like to shed some light on the above important issues.
	Specifically, this article presents a new design objective, i.e, quality of training (QoT), for resource allocation, e.g., bandwidth and power allocation, in multi-modal EIAD systems.
	Specifically, the QoT is defined as the overall perception accuracy (or planning efficiency) of all trained DNNs.
	The QoT metric is monotonically increasing with the communication throughput, but their relationship is nonlinear.
	Hence, the QoT-oriented approach directly maximizing the QoT would  result in fundamentally different designs, compared with that from conventional throughput-oriented approaches.
	Furthermore, despite the QoT-oriented problem being nonlinear and non-smooth, we leverage the accelerated gradient projection (AGP) and dual decomposition methods to optimally solve it in a highly efficient way with low complexity.
	The designed AGP achieves the fastest convergence rate, and the designed dual decomposition yields semi-closed-form solutions.
	The superior performance of the proposed algorithm has been verified by high-fidelity Car Learning to Act (CARLA) simulator \cite{32}.
	
	Finally, we would like to emphasize that QoT-oriented scheduling was studied in \cite{10,47,15,43}, which quantifies the importance of data uploaded from different mobile users by fitting a parametric model to experimental data.
	Our work exploits a similar principle but focuses on domain-specific AD datasets and time-varying channels rather than general-purpose datasets and static channels.
	In addition, EIAD resource allocation was extensively investigated for emerging scenarios such as multicast, federated learning and space-air-ground networks \cite{48,44,49,51}.
	However, these methods ignore the multi-modality issue and thus fail in achieving high QoT under resource constraints.
	In contrast, our method can potentially improve their performance by integrating the QoT-oriented scheduling into the EIAD resource allocation.
	
	\section{System Model and Problem Formulation}
	
	\begin{figure}[!t]
		\centering
		\subfigure[]{
			\label{Fig1a} %% label for second subfigure 63 43
			\includegraphics[height=2in]{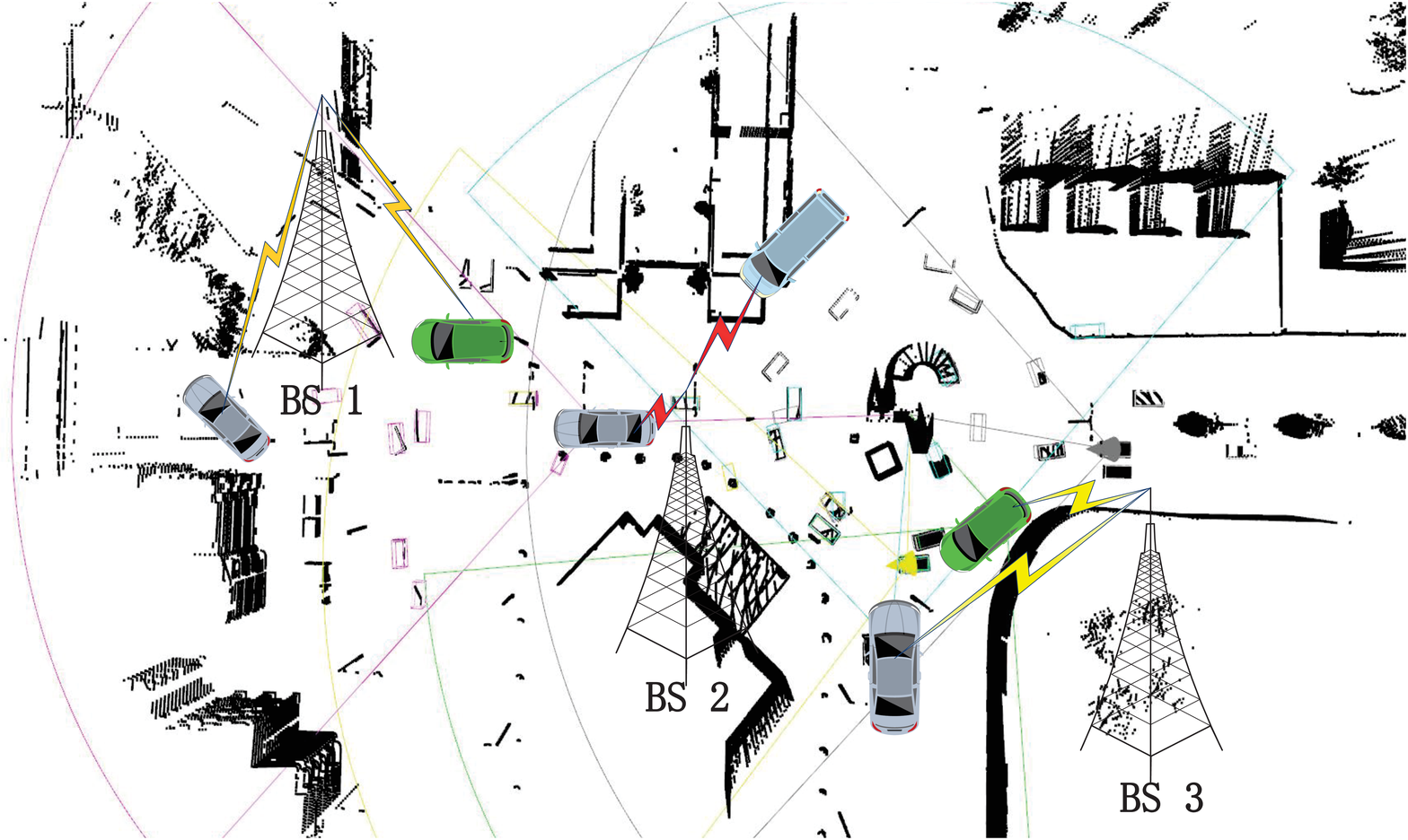}}
		\subfigure[]{
			\label{Fig1b} %% label for second subfigure 63 43
			\includegraphics[height=2in]{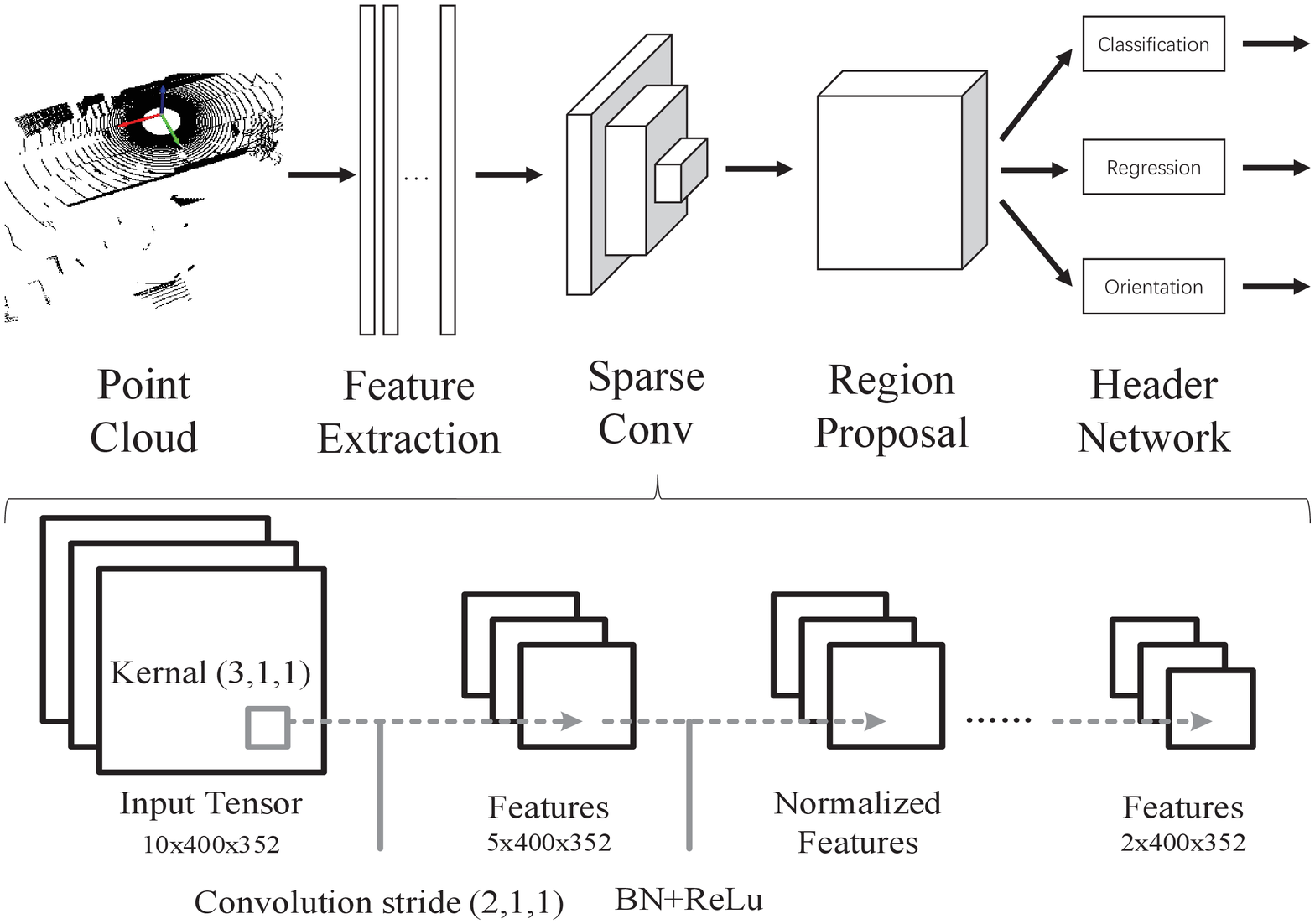}}
		\caption{(a) Illustration of the EIAD system, where the sectors represent different sensing regions of CAVs and the boxes represent the objects detected by the deep learning networks at CAVs;  (b) The structure of SECOND and the diagram of convolution layers with batch normalization and Relu. 
More details can be found at: https://github.com/open-mmlab/OpenPCDet.} 
		\vspace{-7pt}
	\end{figure}

	As shown in Fig.~1(a), we consider the scheduling of data uploading from $K$ connected autonomous vehicles (CAVs) to $L$ edge servers for DNN model training in an EIAD system. In particular, the procedure includes the following 4 stages
\begin{itemize}
\item[1)] \textit{Sensing}: each CAV senses the surrounding environment and stores the sensing data locally; 

\item[2)] \textit{Communication}: the edge server collects datasets from CAVs via uplink transmission; 

\item[3)] \textit{Training}: the edge server annotates the data, trains DNNs with the labeled data, and releases the models to CAVs; 

\item[4)] \textit{Inference}: the performance of the trained DNNs depends on the number of uploaded samples, thus on the bandwidth and power allocated in the communication stage.
\end{itemize}
	The data modality at CAV $k$ is denoted by $M_k$ (with $M_k=1$ representing point-clouds and $M_k=2$ representing images) and the size of its data sample is $D_k$.
    Each point-cloud with $12800$\,Kbits is used to train a sparsely embedded convolutional detection (SECOND) network,   for \textbf{Task 1 (object detection)}.
	Each image with $5600$\, Kbits is used to train a convolution neural network (CNN) for \textbf{Task 2 (weather classification)}. The SECOND architecture is shown in Fig.~1b, and the CNN architecture has four layers with $32 \times 64 \times 28 \times 10$ units.
	Communication time is split into $N$ time slots, where the duration of each time slot is $T$.
	The channels are assumed to be quasi-static during each time slot, and vary in different time slots.
	All channel power gains are assumed to be predictable, since the CAV routes and the traffic map, i.e., mobility pattern,  are known in advance.
	Let $h_{l,k,n}$ denote the uplink channel power gain from CAV $k$ to BS $l$ at time slot $n$.
	The achievable rate between CAV $k$ and BS $l$ at time slot $n$ is
	\begin{align}\label{capacity}
		R_{l,k,n}(w_{l,k,n},q_{l,k,n})=w_{l,k,n}\mathrm{log}_2\left(1+\frac{h_{l,k,n}q_{l,k,n}}{N_0w_{l,k,n}}\right),
	\end{align}
	where $q_{l,k,n}$ and $w_{l,k,n}$ denote the transmit power and bandwidth of CAV $k$ in BS $l$ at time slot $n$, respectively, and $N_0$ denotes the additive white Gaussian noise (AWGN) power spectral density.
	Note that there is no inter-cell interference in the denominator of equation \eqref{capacity}, as adjacent BSs adopt different frequency bands for multiplexing while remote BSs causes random interference that can be included in AWGN. On the other hand, each CAV can be associated to only one BS at a certain time slot.
	Due to limited coverage of BSs and high mobility of CAVs, it is necessary to perform handover during the entire dataset collection procedure.
	To be specific, let $x_{l,k,n}\in\{0,1\}$ with $\sum_{l}x_{l,k,n}=1$ for all $(k,n)$ denote the association state between CAV $k$ to BS $l$ at time slot $n$, where $x_{l,k,n}=1$ represents connection and $x_{l,k,n}=0$ represents disconnection.
	Consequently, the number of samples uploaded for the $k$-th modality is given by
	\begin{align}
		v_{k}=\sum_{l=1}^{L} \sum_{n=1}^{N}{\frac{Tx_{l,k,n}R_{l,k,n}}{ND_k}}.
	\end{align}

	For EIAD perception, the QoT is defined as the perception accuracy or one minus the perception error \cite{zhangjun}.
	Generally, it is difficult to characterize the relationship between the perception error and the number of data samples analytically. Fortunately, based on the research results of \cite{10,50}, this relationship can be approximately characterized by
	$\Psi _{k} \approx a_{k}\, v_{k}^{-b_{k}}$, where $\Psi_{k}$ denotes error rate of the DNN at CAV $k$, and $a_{k}, b_{k}>0 $ are hyper-parameters representing task difficulty. It is further indicated by \cite{10,50} that $a_{k}, b_{k} $ can be obtained from the curve fitting of experimental data.
	Our goal is to optimize the association, bandwidth, and power of all time slots, denoted by
	$\mathcal{X}=\{x_{l,k,n}\}$, $\mathcal{W}=\{w_{l,k,n}\}$, and $\mathcal{Q}=\{q_{l,k,n}\}$, respectively, such that the average perception error is minimized.
	It can be  formulated as the following optimization problem.
	\begin{subequations}
		\label{P1}
		\begin{align}
			\mathcal{P}_0:\,\,\,
			\min_{\mathcal{X},\mathcal{W},\mathcal{Q}
			}~&\sum_{k=1}^K\frac{a_{k}}{K}\left(\sum_{l=1}^L\sum_{n=1}^{N}\frac{T x_{l,k,n}R_{l,k,n}}{ND_{k}}\right)^{-b_{k}}\nonumber  \\
			\textrm{s.t.} ~~ & \frac{1}{N} \sum_{l=1}^K \sum_{n=1}^{N}q_{l,k,n} \leq P_{k}, \quad \forall k, \label{C1} \\
			&\frac{1}{N}\sum_{l=1}^L\sum_{k=1}^{K}\sum_{n=1}^{N}q_{l,k,n} \leq P_{\text{total}}, \label{C2} \\
			&\sum_{k=1}^{K} w_{l,k,n} = B_{\text{total}}, \quad \forall l,n, \label{C3} \\
			&q_{l,k,n} \geq 0, \ w_{l,k,n} \geq 0, \quad \forall l,k,n, \label{C5} \\
			& \sum_{l=1}^Lx_{l,k,n}=1, \quad \forall k,n, \\
			& x_{l,k,n}\in\{0,1\}, \quad \forall l,k,n, \label{xxx},
		\end{align}
	\end{subequations}
	where (3a) is the time slot average individual power constraint for each CAV, (3b) is the time slot average total power constraint for all CAVs, and (3c) is the total bandwidth constraint. The challenges of solving $\mathcal{P}_0$ are two-fold: (1) the discontinuity of the CAV-BS association variables; (2) the curse of dimensionality brought by numerous time slots $N$.
	
	\vspace{-10pt}
	\section{Proposed First-Order Algorithm}
	
	\subsection{Optimal CAV-BS Association}
	
	In order to address challenge (1), we propose the following proposition.
	
	\begin{proposition}\label{prop:opt_association}
		The optimal $\{x_{l,k,n}^*\}$ to $\mathcal{P}_0$ satisfies $x_{l',k',n'}^*=1$ if $h_{l',k',n'}=\max\,\{h_{l,k,n}\}$ and $x_{l',k',n'}^*=0$ otherwise.
	\end{proposition}
	
	\textbf{Proposition 1} can be proved by contradiction.
	Specifically, assume that there exists some $x_{l',k',n'}^*=1$ for $h_{l',k',n'}\neq\max\,\{h_{l,k,n}\}$.
	Then, we can always construct another solution by setting $x_{l',k',n'}=0$ and $x_{ll',kk',nn'}=1$ with $h_{ll',kk',nn'}\neq\max\,\{h_{l,k,n}\}$ while keeping other variables unchanged.
	This would reduces the objective of $\mathcal{P}_0$, which contradicts the optimality of $\{x_{l,k,n}^*\}$.
	Similarly, we can show that the bandwidth $w_{l,k,n}$ and power $q_{l,k,n}$ should be zero if $x_{l,k,n}=0$; otherwise those resources can always be allocated to another link with $x_{l,k,n}=1$ such that the system performance is improved.
	Based on \textbf{Proposition 1}, we substitute $\{x_{l,k,n}=x_{l,k,n}^*\}$ into $\mathcal{P}_0$, which would not change the solution of $\mathcal{P}_0$.
	By setting $\{g_{k,n}=\sum_{l}x_{l,k,n}^*h_{l,k,n}\}$, $\mathcal{U}=\{u_{k,n}:u_{k,n}=\sum_{l}x_{l,k,n}^*w_{l,k,n}\}$, and $\mathcal{P}=\{p_{k,n}:p_{k,n}=\sum_{l}x_{l,k,n}^*q_{l,k,n}\}$,  PROBLEM  $\mathcal{P}_0$ is equivalently transformed into
	\begin{subequations}
		\label{P0}
		\begin{align}
			\mathcal{P}_1:\,\,\,
			\min_{\mathcal{U},\mathcal{P}
			}~&\sum_{k=1}^K\frac{a_{k}}{K}\left(\sum_{n=1}^{N}\frac{T u_{k,n}\mathrm{log}_2\left(1+\frac{g_{k,n}p_{k,n}}{N_0u_{k,n}}\right)}{ND_{k}}\right)^{-b_{k}}\nonumber  \\
			\textrm{s.t.} ~~ & \frac{1}{N} \sum_{n=1}^{N}p_{k,n} \leq P_{k}, \quad \forall k, \label{CC1} \\
			&\frac{1}{N}\sum_{k=1}^{K}\sum_{n=1}^{N}p_{k,n} \leq P_{\text{total}}, \label{CC2} \\
			&\sum_{k=1}^{K} u_{k,n} = B_{\text{total}}, \quad \forall n, \label{CC3} \\
			&p_{k,n} \geq 0, \ u_{k,n} \geq 0, \quad \forall k,n, \label{CC5}.
		\end{align}
	\end{subequations}
	PROBLEM $\mathcal{P}_1$ is convex due to the following reasons:
	\begin{itemize}
		\item All the constraints are linear;
		\item Function ${a_k}(\cdot)^{-b_k}$ is non-increasing and convex;
		\item Function $\frac{Tu_{k,n}}{ND_{k}}\log_{2} \left( 1+\frac{p_{k,n} g_{k,n}}{N_0u_{k,n}} \right)$ is the perspective transform of the logarithm function $\frac{T}{ND_{k}}\log_{2} \left( 1+\frac{p_{k,n} g_{k,n}}{N_0} \right)$ w.r.t $(u_{k,n},p_{k,n})$ and thus jointly concave in both variables;
		\item The objective function is convex by the composition rule.
	\end{itemize}
	
	However, due to challenge (2), if PROBLEM $\mathcal{P}_1$ is solved via the prevalent software CVX, the computation complexity is $\mathcal{O}((KN)^{3.5})$, which is extremely time-consuming for large $N$.
	To overcome the above limitation, we derive a low-complexity optimal solution of large-scale PROBLEM $\mathcal{P}_1$.
	Specifically, since the objective is convex and the constraints of  PROBLEM $\mathcal{P}_1$ are non-coupling, we can optimally solve  PROBLEM  $\mathcal{P}_1$ by alternatively optimizing (AO) bandwidth $\{u_{k,n}\}$ and power allocation $\{p_{k,n}\}$ \cite[Corollary 2]{52}.
	In these two sub-problems, we adopt the first-order AGP and dual decomposition to significantly reduce the computational complexity.
	
	\vspace{-12pt}
	\subsection{Optimal Bandwidth Allocation}
	
	When $\{p_{k,n}=p_{k,n}^\diamond\}$, where $q_{k,n}^\diamond$ denotes the given value of transmission power of the $k$-th CAV in the $n$-th time slot, PROBLEM $\mathcal{P}_1$ is converted to PROBLEM $\mathcal{P}_2$
	\begin{eqnarray}
		\mathcal{P}_2:		&& \!\!\!\!\!\!\!\!\!\! \label{P2}
		\min_{\{u_{l,k}\}}~ \sum_{k=1}^K\frac{a_{k}}{K}\left(\sum_{n=1}^{N}\frac{Tu_{k,n}}{ND_{k}}\log_{2} \left( 1+\frac{p_{k,n}^\diamond g_{k,n}}{N_0u_{k,n}} \right) \right)^{-b_{k}}  \nonumber  \\
		&& \!\!\!\!\!\!\!\!\!\! \,\,\,\textrm{s.t.} ~~~ \eqref{C3}~\text{and}~\eqref{C5}.  \nonumber
	\end{eqnarray}
	According to \cite{21}, the AGP method updates the bandwidth allocation iteratively by the following equation
	\begin{eqnarray}
		\mathbf{U}^{[i+1]} = \Pi_{\mathcal{S}}[\mathbf{Q}^{[i]}-\eta \nabla\Xi(\mathbf{Q}^{[i]})], \label{gd}
	\end{eqnarray}
	where $\mathbf{U}^{[i]}=[\mathbf{u}^{[i]}_1,...,\mathbf{u}^{[i]}_K]^T \in \mathbb{R}^{K\times N}$ is the aggregation of bandwidth allocation of all CAVs at the $i^\text{th}$ iteration. Moreover, $\Pi_{\mathcal{S}}(\cdot) $ is the projection of matrix on set   $\mathcal{S}=\{\mathbf{U}|\mathbf{U}=[\mathbf{u}_1,\ldots,\mathbf{u}_N] (\mathbf{u}_n\in R^{K\times 1}, i=1,\ldots,N); \mathbf{1}^T \mathbf{u}_{n}\leq B_\text{total} ; \mathbf{u}_{n} \succeq 0 \}$, which is elaborated in Appendix A;
	$\eta$ is the step size such that
	$\frac{1}{\eta}\,\mathbf{I} - \nabla^2\Xi_m(\mathbf{U})$ is positive semi-definite;
	$\mathbf{Q}^{[i]}$ is the acceleration point, which is a linear combination of $\mathbf{U}^{[i]}$ and $\mathbf{U}^{[i-1]}$; $\nabla\Xi(\cdot)$ is the gradient of the objective function in  PROBLEM $\mathcal{P}_2$, which is given in the top of next page;
	\begin{figure*}
		\begin{equation}
			[\nabla\Xi(\cdot)] = \frac{\partial \Xi(\mathbf{U})}{\partial u_{k,n}}
			= -\frac{a_{k}b_{k}T}{D_{k}N} \left[\sum_{n=1}^N \frac{u_{k,n}T\log_2(1+\frac{p_{k,n}^\diamond g_{k,n}}{ u_{k,n}N_0})}{D_kN} \right]^{-b_k-1}
			\left[\log_2\left(1+\frac{p_{k,n}^\diamond g_{k,n}}{u_{k,n}N_0}\right)-\frac{1}{\left(\frac{u_{k,n}N_0}{p_{k,n}^\diamond g_{k,n}}+1\right)\ln2} \right]. \nonumber \label{LongEQ}
		\end{equation}
		\hrule
	\end{figure*}
	Thus, we have
	\begin{align}\label{rho}
		\mathbf{Q}^{[i]}=\mathbf{U}^{[i]}+\frac{c^{[i-1]}-1}{c^{[i]}}\left(\mathbf{U}^{[i]}-\mathbf{U}^{[i-1]}\right),
	\end{align}
	where $c^{[i]}$ is a parameter to control the importance of $\mathbf{U}^{[i]}-\mathbf{U}^{[i-1]}$ and is given by
	\begin{align}\label{cm}
		c^{[0]}=1,\quad c^{[i]}=\frac{1}{2}\left(1+\sqrt{1+4\left(c^{[i-1]}\right)^2}\right).
	\end{align}

	\begin{remark} \textbf{Why Acceleration?} We compute the look ahead gradient at the accelerated point $\mathbf{Q}^{[n]}$, by adding some accelerations (the item with $\mathbf{U}^{[i]}-\mathbf{U}^{[i-1]}$ in \eqref{rho}). Nevertheless, to avoid
		over-acceleration, the sequence $c^{[n]}$, which represents how much
		we trust in the acceleration, must be carefully designed as \eqref{cm}. It was proved in \cite{21} that $\mathbf{U}^{[n+1]}$ computed using \eqref{gd}--\eqref{cm} is guaranteed to converge to the optimal solution of  PROBLEM $\mathcal{P}_2$ with an iteration complexity $\mathcal{O}(1/\sqrt{\epsilon})$.
		This iteration complexity achieves the complexity lower bound.
	\end{remark}
	
	\vspace{-10pt}
	\subsection{Optimal Power Allocation}
	
	When the bandwidth allocation vectors is $\{u_{l,k,n}=u_{l,k,n}^\diamond\}$,  where $u_{l,k,n}^\diamond$ denotes the given value of bandwidth of the $k$-th CAV in BS $l$ and the $n$-th time slot, PROBLEM $\mathcal{P}_1$ is converted to PROBLEM $\mathcal{P}_3$ by fixing bandwidth allocation 
	\begin{subequations}
		\label{P3}
		\begin{align}
			\mathcal{P}_3:		\min_{\{\mathbf{p}_{l,k}\}}~&  \sum_{k=1}^K\frac{a_{k}}{K}\left(\sum_{n=1}^{N}\frac{Tu^\diamond_{k,n}}{ND_{k}} \log_{2} \left( 1+\frac{g_{k,n}p_{k,n}}{N_0u_{k,n}^\diamond}\right)\right)^{-b_{k}}\nonumber  \\
			\textrm{s.t.} \ \ & \eqref{C1}, \eqref{C2}, \, \text{and}\,  \eqref{C5}.  \nonumber
		\end{align}
	\end{subequations}
	Applying dual decomposition to PROBLEM  $\mathcal{P}_3$ yields 
	\begin{eqnarray}
		&&\!\!\!\!\!\!\! \!\!\!\!\!\!\!  \mathcal{D}(\mathcal{P}_3): \max_{\lambda\geq 0}
		\min_{\{\mathbf{p}_k \in \mathcal{G}_k\}}~   \sum_{k=1}^K\frac{a_{k}}{K}\left(\sum_{n=1}^{N}\frac{Tu^\diamond_{k,n}}{ND_{k}} \log_{2} \left( 1 + \right. \right. \nonumber  \\
		&&\!\!\!\!\!\!\!  \!\!\!\!\!\!\!  \left. \left. \frac{g_{k,n}p_{k,n}}{N_0u_{k,n}^\diamond} \right)\right)^{-b_{k}} + \lambda \left(\frac{1}{N}  \sum_{k=1}^{K}\sum_{n=1}^{N}p_{k,n}-P_{\text{total}}\right),
	\end{eqnarray}
	where
	$\mathcal{G}_k=\left\{\mathbf{p}_k:
	\frac{1}{N} \sum_{n=1}^{N}p_{k,n} \leq P_{k}, \,\, p_{k,n} \geq 0\right\}.
	$
	The dual of PROBLEM $\mathcal{P}_3$ is a bilevel optimization problem, where the outer problem is an unconstrained nonsmooth maximization problem and the inner problem is a constrained but decomposable problem.
	In the outer problem, the dual variable $\lambda$ can be updated via the sub-gradient descent method as
	\begin{align}
		&\lambda^{[i+1]}=\lambda^{[i]}+\xi\left(\frac{1}{N} \sum_{k=1}^{K}\sum_{n=1}^{N}p_{k,n}^{[i]}-P_{\text{total}}\right), \label{lambda}
	\end{align}
	where $\lambda^{[i]}$ and $p_{k,n}^{[i]}$ are  dual variable $\lambda$ and power allocation $p_{k,n}$ of  the $i^\text{th}$ iteration, respectively, and $\xi$ is the step size.
	At the $i$-th iteration, the inner problem for fixed $\lambda^{[i]}$ can be equivalently decomposed into $K$ sub-problems, given  by
	\begin{align}
		\mathcal{P}_4^{[i]}(k): \min_{\mathbf{p}_k^{[i]}\in\mathcal{G}_k}&\frac{a_{k}}{K}\left(\sum_{n=1}^{N}\frac{Tu^\diamond_{k,n}}{ND_{k}} \log_{2} \left( 1+\frac{g_{k,n}p_{k,n}^{[i]}}{N_0u_{k,n}^\diamond} \right)\right)^{-b_{k}} \nonumber  \\
		&+\lambda^{[i]}\frac{1}{N}\sum_{n=1}^{N}p_{k,n}^{[i]}, \, \,k=1,\cdots,K.
		\label{pki}
	\end{align}
	Define  $t_k=\frac{1}{N}\sum_{n=1}^{N}p_{k,n}^{[i]}$ as a slack variable. PROBLEM $\mathcal{P}_4^{[i]}(k)$ can be equivalently written as
	\begin{align}
		\min_{\left\{
			\mathbf{p}_k^{[i]},
			t_k
			\right\}} &\frac{a_{k}}{K}\left(\sum_{n=1}^{N}\frac{Tu^\diamond_{k,n}}{ND_{k}}\log_{2} \left( 1+\frac{g_{k,n}p_{k,n}^{[i]}}{N_0u_{k,n}^\diamond} \right)\right)^{-b_{k}}+\lambda^{[i]}t_k\nonumber\\
		\mathrm{s.t.}~ & \eqref{C1}, \eqref{C5}, \nonumber \\
		& \frac{1}{N}  \sum_{n=1}^{N}p_{k,n}^{[i]}=t_k.  \label{AppA1}
	\end{align} We have the following proposition on the optimal solution of  Problem  \eqref{AppA1}.
	\begin{proposition}\label{prop:opt_power}
		Given $t_k$, the optimal $\mathbf p_k^*$ to Problem \eqref{AppA1} is
		\begin{align}\label{p*}
			{p_{k,n}^{[i],*}}(\mu)= \left[ \frac{Tu_{k,n}^\diamond}{\mu ND_k\ln2}-\frac{N_0 u_{k,n}^\diamond}{g_{k,n}} \right]^+, \forall k, n,
		\end{align}
		where $\mu>0$ satisfies $ \sum_{n=1}^N{p_{k,n}^{[i],*}}(\mu)=N\mathrm{min}\left(P_k,t_k\right)$.
	\end{proposition}
	\begin{proof}
		Please refer to Appendix B.
	\end{proof}
	
	\begin{remark} Proposition 1 indicates that the optimal $\mathbf p_k^*$ can be found via one-dimensional search over $t_k$.
		Since $t_k\in[0,P_{k}]$ and the objective function is  uni-modal w.r.t. $t_k$, the optimal $t_k^*$ in (13) can be found by bisection search within $[0,P_{k}]$.
		The iteration complexity of bisection is $\mathcal{O}(\mathrm{log}(1/\epsilon))$.
	\end{remark}
	
	\begin{algorithm}[b]
		\caption{Proposed First-Order Algorithm for Solving  $\mathcal{P}_1$}
		\begin{algorithmic}[1]
			\State Initialize $\eta = 10^4$ and $\xi = 10^{-3}$.
			\State \textbf{Repeat}:
			\State \quad Set $i=1$ and $c^{[0]} = 1$.
			\State \quad \textbf{Repeat}:
			\State \quad \quad Calculate $c^{[i]}$, $\textbf{Q}^{[i]}$, $\textbf{U}^{[i+1]}$ based on \eqref{gd}--\eqref{cm}.
			\State \quad \quad Update $i\leftarrow i+1$
			\State \quad \textbf{Until}: The stop criterion is satisfied.
			\State \quad Set $i=0$ and $\lambda^{[0]} = 0$.
			\State \quad \textbf{Repeat}:
			\State \quad \quad Calculate ${p_{k,n}^{[i],*}}$ based on \eqref{p*}.
			\State \quad \quad Update $\lambda^{[i+1]}$ based on \eqref{lambda}.
			\State \quad \quad Update $i\leftarrow i+1$
			\State \quad \textbf{Until}: The stop criterion is satisfied.
			\State \textbf{Until}: The stop criterion is satisfied.
		\end{algorithmic}
	\end{algorithm}

	\vspace{-6pt}
	\subsection{Complexity Analysis}
	
	The entire procedure of the proposed method is summarized in Algorithm 1.
	It can be seen that Algorithm 1 involves two levels of iterations.
	In the outer-level AO iteration, the AGP method is first adopted to solve $\mathcal{P}_2$, which executes \eqref{gd}--\eqref{cm} iteratively.
	The computation is dominated by equation \eqref{gd}, which requires a complexity of $\mathcal{O}(KN^2)$ (computing each element in $\nabla\Xi(\mathbf{Q}^{[i]})]$ needs a complexity of $\mathcal{O}(N)$ and there are $KN$ elements).
	Consequently, with $\mathcal{O}(1/\sqrt{\epsilon})$ iterations, the AGP method costs a complexity of $\mathcal{O}(KN^2/\sqrt{\epsilon})$.
	Then, to solve $\mathcal{P}_3$, dual decomposition is adopted, which executes \eqref{p*} for all $(k,n)$ and \eqref{lambda} iteratively.
	The computation cost of \eqref{p*} for all $(k,n)$ is given by $\mathcal{O}(KN)$.
	Therefore, with $\mathcal{O}(1/\epsilon)$ iterations for sub-gradient update and $\mathcal{O}(\mathrm{log}(1/\epsilon))$ iterations for bisection search, the dual decomposition method requires a computation complexity of $\mathcal{O}(\mathrm{log}(1/\epsilon)KN/\epsilon)$.
	In summary, the total complexity of Algorithm 1 is given by
	$\mathcal{O}(\mathrm{ITER}(KN^2/\sqrt{\epsilon}+\mathrm{log}(1/\epsilon)KN/\epsilon))$, where $\mathrm{ITER}$ is the number iterations for AO to converge.

	\section{Simulation Results}
	
	\begin{figure*}[t]
		\centering
		\subfigure[]{
			\label{Fig5} %% label for second subfigure 63 43
			\includegraphics[height=1.34in]{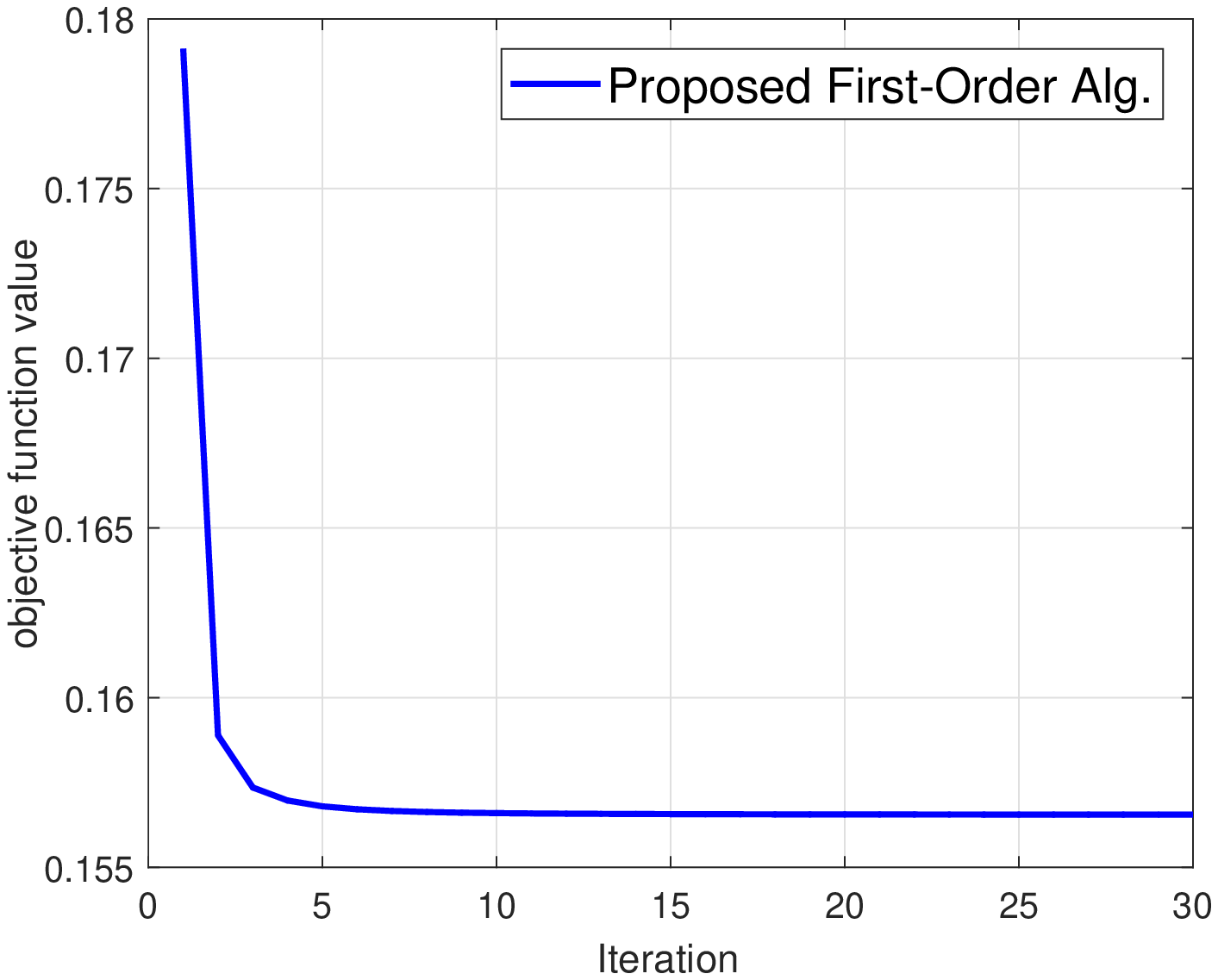}}
		\subfigure[]{
			\label{Fig6} %% label for second subfigure 63 43
			\includegraphics[height=1.34in]{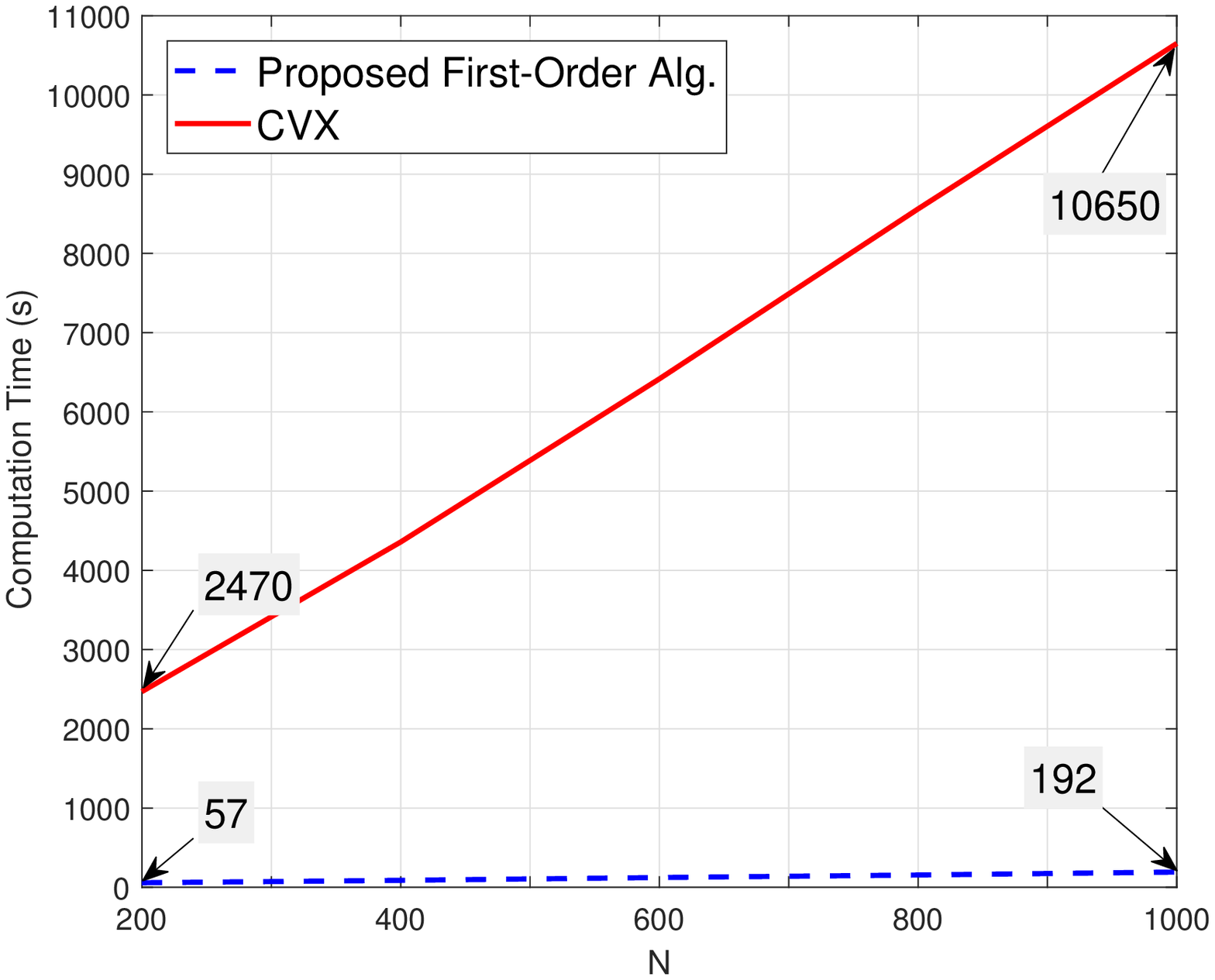}}
		\subfigure[]{
			\label{Fig7} %% label for second subfigure 63 43
			\includegraphics[height=1.34in]{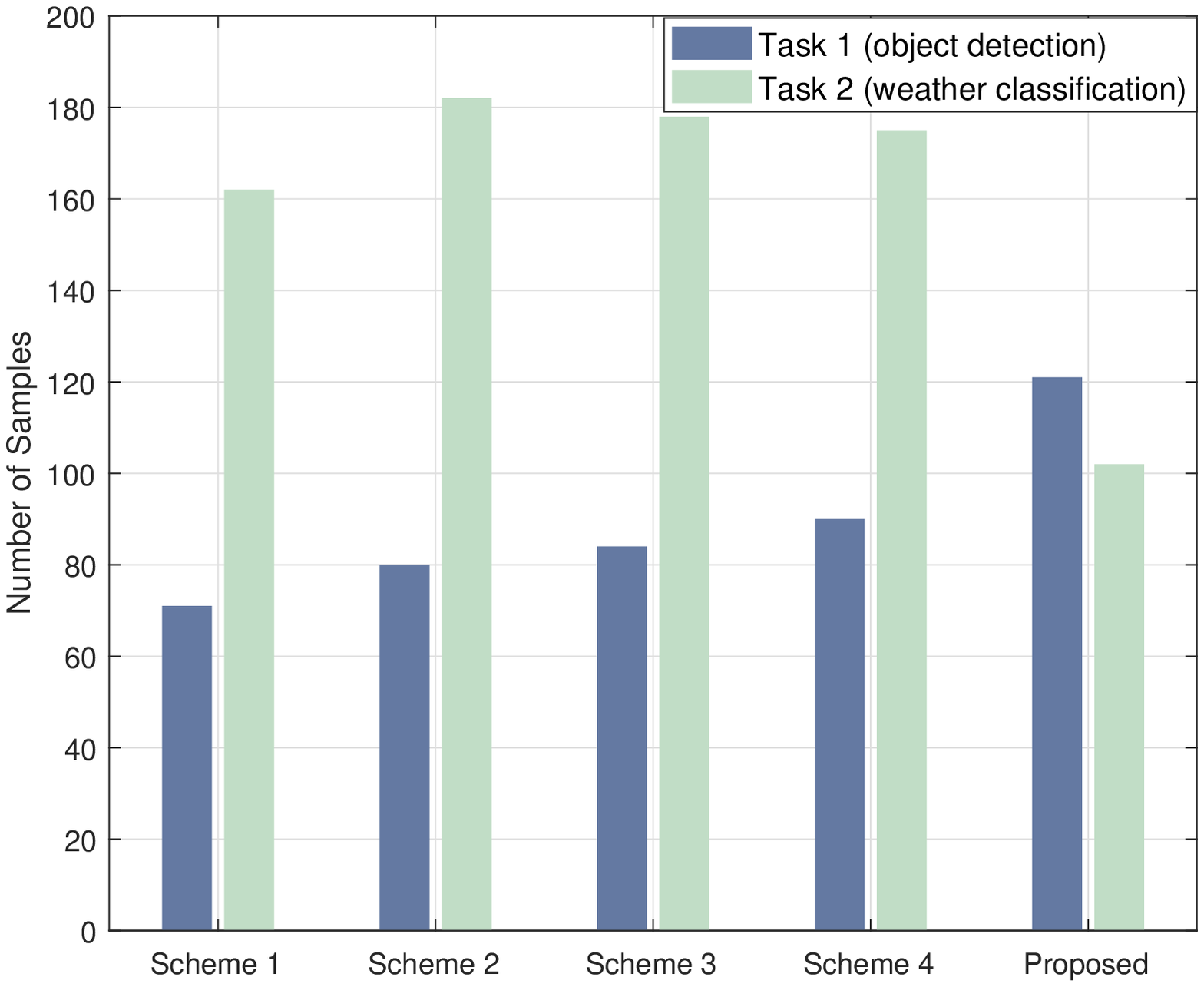}}
		\subfigure[]{
			\label{Fig8} %% label for second subfigure 63 43
			\includegraphics[height=1.34in]{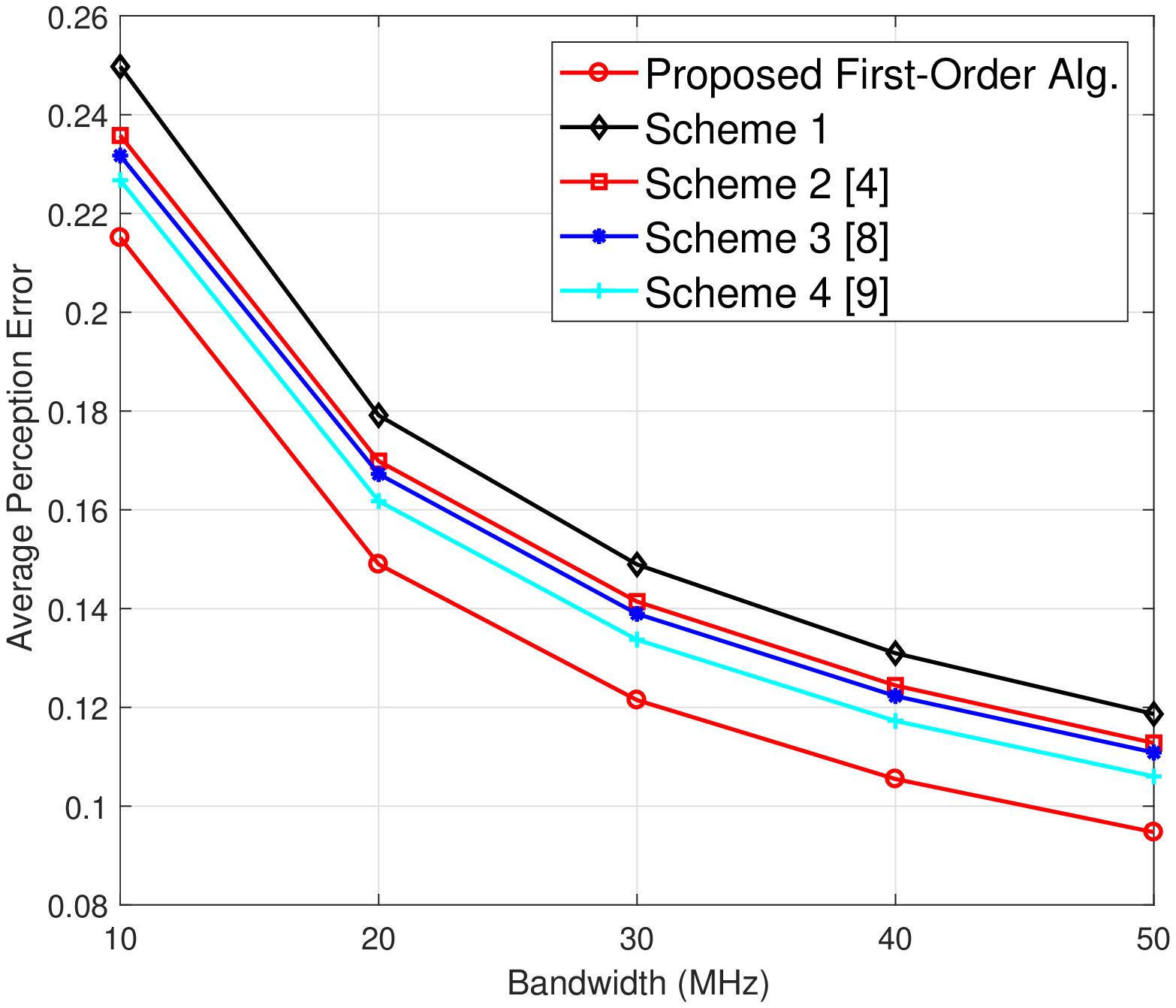}}
		\vspace{-0.3cm}
		\caption{(a) The objective value versus the number of iterations; (b) Comparison of the computation time; (c) Comparison of the number of collected samples; (d) Average perception errors of both asks under different bandwidths.}
		\label{Simulation_2} %% label for entire figure
		\vspace{-5pt}
	\end{figure*}
	
	\begin{figure*}[t]
		\centering
		\subfigure[]{
			\label{Fig9}
			\includegraphics[width=0.49\textwidth]{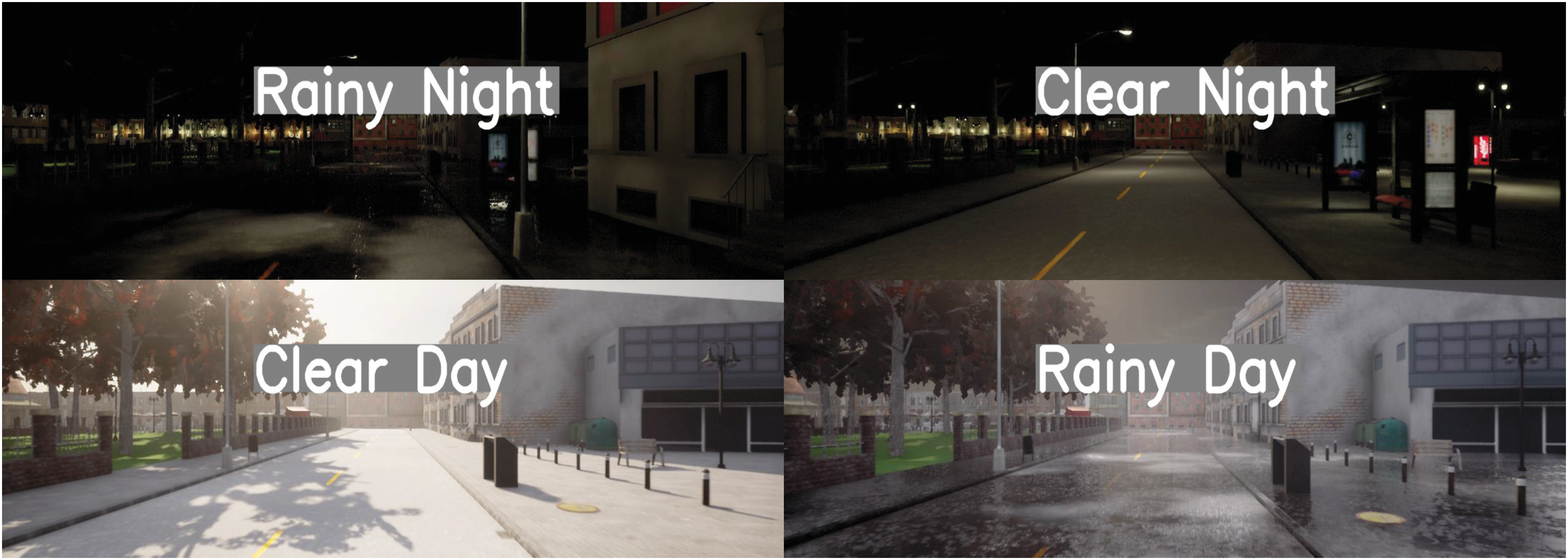}}
		\subfigure[]{
			\label{Fig10}
			\includegraphics[width=0.49\textwidth]{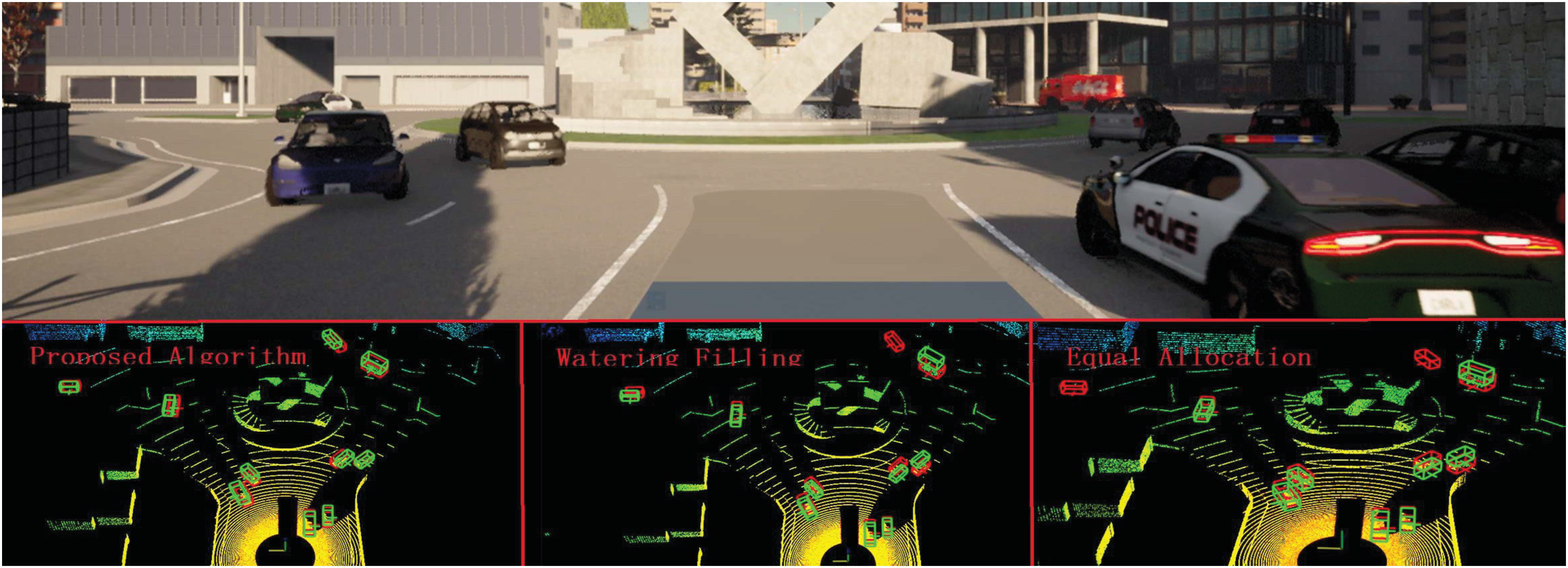}}
		\vspace{-0.3cm}
		\caption{(a) The weather detection task implemented in CARLA; (b) The object detection task implemented in CARLA. The red box represents the ground truth, and the green box represents the detected result.}
		\label{Simulation_3} %% label for entire figure
		\vspace{-15pt}
	\end{figure*}
	
	The simulations were done on the CARLA simulator \cite{32,flcav} with $K=2$, where
	the first CAV with a 64-line LiDAR on top of the car is a point-cloud data collector, while the second CAV with an RGB camera in front of the car is an image data collector.
	The Adam optimizer is adopted for training SECOND and CNN.
	The communication parameters are given by $N=1000$, $B_\text{total} = 20$\,MHz, $P_1 = P_2 = 1$\,W, $P_{\text{total}}=2$\,W, $N_0=-110$\,dBm/Hz.
	We simulate $L = 10$ BSs and the CAV-BS distance is generated randomly from $5$\,m to $150$\,m.
	The channels are generated by using a distance-dependent path-loss model with $30$\,dB loss at a unit distance of $1$\,m.
	
	Fig.~\ref{Fig5} shows the convergence behaviour of the proposed algorithm.
	It can be seen that the algorithm converges very fast within 10 iterations.
	Fig.~\ref{Fig6} compares the computation time of the proposed algorithm and CVX on a desktop with I7-7700 3.6GHz CPU and 64G RAM.
	The proposed algorithm significantly reduces computation time.
	Notably, when $N=1000$, the proposed first-order algorithm reduces the computation time of CVX (i.e., interior-point method) by $98.2\%$, and the gain  increases with the number of time slot $N$.
	
	Next, we compare the proposed algorithm with benchmark schemes as follows:
	1) \textbf{Scheme 1}: equally allocating bandwidth and power across all CAVs;
	2) \textbf{Scheme 2}: maximizing the total communication throughput via water-filling \cite{48};
	3) \textbf{Scheme 3}: QoT-oriented power optimization with equal bandwidth allocation \cite{10};
	4) \textbf{Scheme 4}: QoT-oriented bandwidth and power allocation ignoring time-varying channels \cite{47}.
	Figs. \ref{Fig7} and \ref{Fig8} compare the number of uploaded samples and the average perception error of the proposed algorithm with those of the benchmark schemes, respectively.
	It can be seen from Fig.~\ref{Fig7} that the proposed algorithm leads to a more balanced sample allocation between tasks 1 and 2.
	This is because the proposed algorithm simultaneously exploits the properties of multi-modal datasets and time-varying channels.
	As such, the proposed algorithm achieves a significantly smaller perception error, i.e., a higher QoT, than those of other benchmarks as shown in Fig.~\ref{Fig8}.
	
	Compared with Scheme 1, the proposed scheme reduces the perception error by $3\%$, which implies that resource allocation is crucial to the EIAD systems.
	Moreover, Scheme 2 leads to the second-worst performance among all the simulated schemes, meaning that the objective function have a more significant impact on EIAD than other factors such as the choice of design variables and the input channels.
	Finally, by comparing the proposed method with Schemes 4 and 3, we find that ignoring the time-varying feature of wireless channels would degrade the system performance inevitably.
	
	The simulation results of Figs. \ref{Fig7} and \ref{Fig8} are further visualized in Fig. \ref{Fig9} and Fig. \ref{Fig10}.
	In particular, no matter which algorithm is chosen, the trained CNNs always distinguish different weathers.
	This is because task 2 has a fast learning progress, and tens of images are enough for realizing accurate perception.
	On the other hand, the SECOND trained with the proposed algorithm successfully detects objects on the road.
	In contrast, other schemes yield missing or inaccurate detection results due to the insufficient number of point clouds.
	This is because the proposed algorithm automatically allocates more resources to task 1, which has a more significant learning curve, for QoT maximization.
	
	\vspace{-6pt}
	\section{Conclusion}
	
	This article has studied the large-scale bandwidth and power allocation problem in EIAD.
	A first-order accelerated algorithm with linear complexity has been proposed.
	The proposed algorithm achieved a smaller perception error than the state-of-the-art, and a lower complexity than interior point method.

		\begin{appendices}

		\section{Projection of $\Pi_{\mathcal{S}}(\mathbf{X})$}
		$\Pi_{\mathcal{S}}(\mathbf{X})$ is a process of projecting a given point $\mathbf{X} =[\mathbf{x}_1,\ldots,\mathbf{x}_N] (\mathbf{x}_n\in R^{K\times 1}, n=1,\ldots,N)$
		onto set $\mathcal{S}$.
		This problem is equivalent to
		\begin{eqnarray}
		&& \mathcal{O}:\min_{\mathbf{U}} ||\mathbf{X}-\mathbf{U}||_2^2=\sum_{n=1}^N ||\mathbf{x_n}-\mathbf{u}_n||_2^2  \nonumber \\
		&& \mathrm{s.t} \quad \mathbf{1}^T \mathbf{u}_n = B_\text{total}, ~ \mathbf{u}_{n} \succeq \textbf{0}, \quad \forall n. \nonumber
		\end{eqnarray}
		Since the variable $\mathbf{U}$ can
		be partitioned into subvectors $\mathbf{u}_1, \cdots, \mathbf{u}_N$, the objective is a sum of functions of $\mathbf{u}_n$, $n=1,\cdots,N$,
		and each constraint involves only variables from one of the subvectors $\mathbf{u}_n$.
		Then, we can solve each problem involving $\mathbf{u}_n$ separately, and re-assemble the solution $\mathbf{U}$ as $[\mathbf{u}_1,\ldots,\mathbf{u}_N]$.
		The $n$-th sub-projection problem is written as
		\begin{eqnarray}
		\mathcal{O}^{(n)}:\min_{\mathbf{u}_n} ||\mathbf{x}_n-\mathbf{u}_n||_2^2 \quad \mathrm{s.t} \quad \mathbf{1}^T \mathbf{u}_n = B_\text{total},\quad \mathbf{u}_n\succeq \textbf{0}. \nonumber
		\end{eqnarray}
		The solution $\mathbf{u}_n$ to problem  $\mathcal{O}^{(n)}$ can be calculated according to \cite[Proposition~2.2]{23}:
		\begin{align}
		&
		\mathbf{u}_n=
		\left[\mathbf{x}_n-\frac{\mathop{\sum}_{l=1}^\delta z_l-B_\text{total}}{\delta}\right]^+, \label{proj}
		\end{align}
		where $\mathbf{z}$ is a permutation of $\mathbf{x}_n$ such that $z_1\geq \cdots \geq z_N$, and
		\begin{align}
		\delta=\mathop{\mathrm{max}}_{m\in\{1,\cdots,N\}}~\left\{m:
		\frac{\sum_{l=1}^mz_l-B_\text{total}}{m}<z_m
		\right\}.
		\end{align}

		%\section{Proof of Theorem \ref{prop:opt_power}}\label{app:opt_power}
		
		\section{Proof of Proposition 2}
		
		For solving Problem \eqref{AppA1}, $\sum_{n=1}^{N}p_{k,n}^{[i]}=t_k$ is relaxed to  $\frac{1}{N} \sum_{n=1}^{N}p_{k,n}^{[i]}\leq t_k$, which does not change the problem.
		This is because the objective function is monotonically decreasing in $p_{k,n}^{[i]}$ and a larger  $\frac{1}{N} \sum_{n=1}^{N}p_{k,n}^{[i]}$ always reduces the objective function.
		Therefore, optimal $\mathbf{p}_k^*,t_k^*$ to Problem \eqref{AppA1} always activate   $\frac{1}{N} \sum_{n=1}^{N}p_{k,n}^{[i]}\leq t_k$.
		As such, Problem \eqref{AppA1} can be re-written as
		\begin{align}
		\min_{\{\mathbf{p}_k^{[i]},t_k\}}~&\frac{a_{k}}{K}\left(\sum_{n=1}^{N}\frac{Tu^\diamond_{k,n}}{ND_{k}} \log_{2} \left( 1+\frac{g_{k,n}p_{k,n}^{[i]}}{N_0u_{k,n}^\diamond} \right)\right)^{-b_{k}}  +\lambda^{[i]}t_k\nonumber\\
		\mathrm{s.t.}~~~&\eqref{C5}, \quad
		\frac{1}{N}\sum_{n=1}^{N}p_{k,n}^{[i]} \leq \mathrm{min}\left(P_{k},t_k\right). \label{AppA2}
		\end{align}
		For each fixed $t_k=t_k^\diamond$, the objective function is only related to $\mathbf{p}_k^{[i]}$.
		Furthermore, due to monotonic decreasing property of $a_kx^{-b_k}$, Problem \eqref{AppA2} can be equivalently transformed into
		\label{prob:equiv_p}
		\begin{align}
		\min_{p_{k,n}^{[i]}}\quad &-\sum_{n=1}^N \frac{Tu_{k,n}^{\diamond}}{ND_{k}}
		\log_2\left(1+\frac{p_{k,n}^{[i]}g_{k,n}}{N_0u_{k,n}^{\diamond}}\right)\nonumber\\
		\textrm{s.t.} \quad& \eqref{C5},\quad \sum_{n=1}^Np_{k,n}^{[i]}\leq N\mathrm{min}\left(P_{k},t_k\right),
		\end{align}
		It can be seen that Problem (18) is a convex problem. Thus, its optimal solution ${p_{k,n}^{[i],*}}$ can be obtained by KKT conditions.
		%The Lagrange function of (18) is given by
		%\begin{equation} \label{Lagrangian}
		%	\begin{split}
		%		\mathcal{L} &= -\sum_{n=1}^N \frac{Tu_{k,n}^\diamond}{ND_k} \log_2\Big(1+\frac{p_{k,n}^{[i]}g_{k,n}}{N_0u_{k,n}^\diamond}\Big)\\
		%		& +\mu\Big(\sum_{n=1}^Np_{k,n}^{[i]}-N\min(P_k,t)\Big)-\sum_{n=1}^N \alpha_n p_{k,n}^{[i]},
		%	\end{split}
		%\end{equation}
		%where $\mu \text{ and } \alpha_n(n=1,...,N) \ge 0$ are the Lagrange multiplier.
		Particularly, according to the  stationarity condition, we have
		\begin{equation}
		-\frac{Tu_{k,n}^\diamond}{ND_k ln2}\frac{g_{k,n}/N_0u_{k,n}^\diamond}{1+{p_{k,n}^{[i],*}} g_{k,n}/N_0 u_{k,n}^\diamond}+\mu - \alpha_n = 0.
		\end{equation}
		The solution is either ${p_{k,n}^{[i],*}}=0$ or  ${p_{k,n}^{[i],*}} > 0 $. Moreover, if ${p_{k,n}^{[i],*}} \neq 0$,  according to the complementary slackness, $\alpha_n = 0$. Thus, the optimal ${p_{k,n}^{[i],*}}$ is given by
		\begin{equation}
		{p_{k,n}^{[i],*}}(\mu) = \left[ \frac{Tu_{k,n}^\diamond}{\mu ND_kln2}-\frac{N_0 u_{k,n}^\diamond}{g_{k,n}} \right]^+.
		\end{equation}
		Moreover, $\mu \neq 0$ (otherwise, ${p_{k,n}^{[i],*}}(\mu) \rightarrow +\infty $). Thus, according to the complementary slackness, it holds that $\sum_{n=1}^N {p_{k,n}^{[i],*}}(\mu)=N\min(P_k,t_k)$. This ends the proof.
	\end{appendices}
	
	\vspace{-7pt}
	\bibliographystyle{IEEEtran}
	\bibliography{perception}
	
\end{document}